\documentclass[a4paper,conference]{IEEEtran}
\IEEEoverridecommandlockouts
\usepackage{epsf}
\usepackage{graphicx}
\usepackage[outdir=./]{epstopdf}
\usepackage{amsmath,bm}
\usepackage{amssymb}
\usepackage{epsfig,latexsym,amsmath,epsf,amssymb,amsfonts}
\usepackage{verbatim}

\usepackage[english]{babel}
\usepackage[utf8]{inputenc}
\usepackage{algorithm}
\usepackage[noend]{algpseudocode}

\usepackage{epstopdf}
\usepackage{amsthm}
\usepackage{url}
\usepackage[noadjust]{cite}
\usepackage{caption}
\usepackage{subcaption}

\newtheorem{theorem}{Theorem}

\usepackage{authblk}
\usepackage{amscd,mathrsfs}
\usepackage[numbers,sort]{natbib}
\usepackage[withbib,all]{authorindex}
\usepackage{esint}
\usepackage{paralist}
\usepackage{color}
\usepackage{bm}
\usepackage{amsthm}
\usepackage{mathtools}
\usepackage{mathptmx} % assumes new font selection scheme installed
\usepackage{times} % assumes new font selection scheme installed
\usepackage{amssymb}
\usepackage{dsfont}

\begin{document}

\title{Finite Horizon Throughput Maximization for a Wirelessly Powered Device over a Time Varying Channel}
%\vspace{-1cm}
\author[1]{Mehdi Salehi Heydar Abad \thanks{This work was in part supported by EC H2020-MSCA-RISE-2015 programme under grant number 690893.}}
\author[1]{Ozgur Ercetin}

\affil[1]{Faculty of Engineering
and Natural Sciences, Sabanci University}

\affil[ ]{\textit{\{mehdis,oercetin\}@sabanciuniv.edu}}
\maketitle

\newtheorem{lemma}{Lemma}
\newtheorem{corollary}{Corollary}
\thispagestyle{empty}

\begin{abstract}
%Mobile edge computing (MEC) has the promise to bring the resources of the cloud closer to the clients.  A MEC server has the potential not only to provide computational resources, but also to provide energy by wireless power transfer (WPT). 

In this work, we consider an energy harvesting device (EHD) served by an access point with a single antenna that is used for both wireless power transfer (WPT) and data transfer. The objective is to maximize the expected throughput of the EHD over a finite horizon when the channel state information is only available causally.  The EHD is energized by WPT for a certain duration, which is subject to optimization, and then, EHD transmits its information bits to the AP until the end of the time horizon by employing optimal dynamic power allocation.
The joint optimization problem is modeled as a dynamic programming problem.  Based on the characteristic of the problem, we prove that a time dependent threshold type structure exists for the optimal WPT duration, and we obtain closed form solution to the dynamic power allocation in the uplink period. 
%More importantly, the distributed policy's performance does not deteriorate with respect to the number of EHDs.
\end{abstract}
\section{Introduction}

IoT devices are typically powered either by finite capacity batteries or by energy harvested from the ambient energy resources. In particular, wireless power transfer (WPT) is considered as a promising technology, where RF signals are utilized as a mean to transfer power to energy harvesting IoT devices (EHDs) \cite{6951347}. In this work, we investigate a system where an access point (AP) periodically collects information from an EHD as shown in Figure \ref{fig:offload}. The AP first performs WPT to replenish the battery of the EHD for a duration that is subjected to optimization. Once this energy harvesting (EH) period ends, information transmission (IT) period begins, where the EHD transmits its data to the AP by {\em dynamically} adjusting its transmission power until the deadline. The condition of the channel varies randomly over time so that the amount of energy transferred from the AP to the EHD as well as the bits transmitted from the EHD to the AP varies randomly over time. We aim to maximize the expected throughput by the deadline. 

%Examples of such studies, but not limited to, for different scenarios under various assumptions are \cite{Wang17ad,7959896,7417552,7997477,7572018}. 

There is a recent interest in developing algorithms for efficient operation of networks with wireless powered devices. The authors in \cite{7066975}, consider a similar problem wherein a transmitter uses WPT to charge the battery of a receiver for a certain duration and then receives data over a finite horizon. However, they only considered a system model where the channel remains static over the horizon. In \cite{6678102}, an AP transmits energy to multiple receivers for a certain duration and then collects data by employing time division multiple access. The energy transfer duration and access times are optimized to maximize the throughput of the network. For a full-duplex (FD) setting where the energy transfer and data transfer operate simultaneously, \cite{6907966} maximizes the sum throughput of the network by optimizing power and time allocation. In a finite horizon, \cite{7996351} studies the throughput maximization where the AP employs non-orthogonal multiple access (NOMA) to simultaneously receive and decode interfering information. In \cite{7332956}, multiple devices harvest energy from a dedicated power station while communicating with a  separate base station to convey their data.  Time allocation and power control in the downlink and the uplink are optimized for maximizing the system energy efficiency. In \cite{7676282} the problem of long term throughput maximization for two nodes in a WPT scenario is studied to optimize the energy transfer, uplink access times, and power allocation using a Markov decision process (MDP). 

%We show that the optimal energy transmission time follows a time dependent threshold type structure. We also find an optimal power allocation strategy that maximizes the the finite horizon throughput. 
All of the aforementioned works assume that the channel state stays constant during the system operation which is not true in general. In this paper, we consider a realistic channel model where the wireless channel changes randomly during both the EH and IT periods. Also note that many of the earlier works on  finite horizon throughput maximization problem considered a dynamic program (DP) formulation and attempted to solve it \emph{numerically} offline.  This solution is suggested to be later stored in the devices as a {\em look-up table}. However, the solution of DP is usually computationally expensive, and requires a large memory space to store, which may be prohibitive for resource-constrained EHDs. Moreover,  calculating and disseminating the optimal look-up table in a network with large number of EHDs is inherently challenging and it introduces a large overhead \cite{7736112}. Hence, unlike previous works, we obtain the structure of the optimal policy and show that the optimal duration of EH period has a time-varying threshold structure. We derive analytical expressions for evaluating the time-dependent threshold. Finally, we find closed form expressions for the optimal power allocation in the IT period based on the remaining time and energy of the EHD. The main contributions of this paper can be summarized as follows:
\begin{itemize}
\item We formulate a finite horizon throughput maximization problem with joint time and power allocation by considering the random behavior of the channel in the horizon. 
\item We find a time-dependent threshold structure that dictates the optimal duration of the EH period. We give a framework to obtain the values of the time-dependent threshold.
\item In the IT period, we derive analytical expressions for the transmission power based on the residual time and energy, while the channel state information is known only causally.
\end{itemize}
\section{System Model}\label{SystemModel}

We consider a WPT scenario consisting of a single EHD and a single AP as shown in Figure \ref{fig:offload}. Time is slotted, with $t = 1,2,\ldots, T$ and a {\em time frame} has a length of $T$ slots.  Let $T$ be a prespecified parameter determined by the network administrator according to the needs of the application. Time frame is split into energy harvesting (EH) and information transmission (IT) periods. The AP is responsible for replenishing the energy of the EHD via RF transmissions in the EH period, and collecting information bits from the EHD in the IT period. The EH and IT periods are non-overlapping in time, assuming a half-duplex transmission scenario. 
 \begin{figure}[h]
  \centering
    \includegraphics[scale=.45]{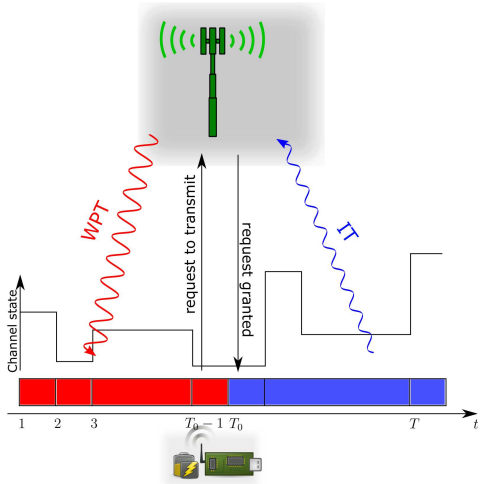}
		  \caption{System model.}
			\label{fig:offload}
\end{figure}
The wireless channel is modeled as a multi state independent and identically distributed (iid) random process with $N$ levels.  The channel gain remains constant for a duration of a  time slot but changes  randomly from one time slot to another.  Let $g(t)\in \{g_1,\ldots,g_N\}$ be the channel power gain at slot $t$. We set $\mathds{P}(g(t)=g_n)=q_n$\footnote{Note that $g_n$'s can be obtained by discretizing a continuous time channel process.}. The EHD has causal channel state information (CSI) and only during the IT period. 
%hile harvesting energy, it is blind to the channel conditions.

%The EHD is required to finish a certain task by a deadline of $T$ time slots. If the task is not finished within the deadline, it will be dropped. Let us denote by $L$ the number of bits of input data necessary for the computation. Note that in MEC scenarios $L$ and $T$ are not fixed, since different tasks with different input sizes and deadlines may be requested from the EHD. 

In the EH period, the EHD first recharges its battery for a duration of $T_0-1$ slots, which is an optimization parameter, and then, utilizes the harvested energy to deliver its bits to the AP in the subsequent IT period from $t=T_0$ to $T$ slots. The AP transmits a power beacon of $P$ watts over the wireless channel for a duration of $T_0-1$ time slots. We assume a time slot normalized set-up, and thus, we will refer to power and energy interchangeably. The energy state of the EHD at time slot $t$ is denoted by $E(t)$.  

 \begin{figure}[h]
  \centering
    \includegraphics[scale=.47]{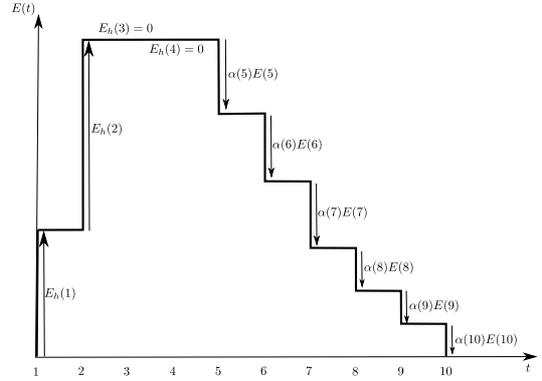}
		  \caption{An illustrative example of battery evolution, $E(t)$, where $E_h(t)$ denotes the amount of harvested energy at time $t$ and $T=10$ . The EHD harvests energy until $t=4$ and then starts transmitting to the AP at $t=5$ by utilizing $\alpha(t)$ portion of its available energy. }
			\label{fig:system}
\end{figure}

At the beginning of each slot, the EHD has the opportunity to inform the AP to stop the EH period and begin IT period.  Let time slot $T_0$ be the time slot when the EHD informs the AP. In order to develop a tractable analytical solution, we assume an empirical transmission energy model as in \cite{6574874,4674675}. Specifically, the amount of energy required to transmit $l$ bits in time slot $t$ is given by:
\begin{align}
\mathcal{E}(l,g(t)) = \frac{\lambda l^m}{ g(t)},
\end{align}
where $\lambda$ denotes the energy coefficient dictating the effects of bandwidth and noise power,
and $m > 1$ is the monomial order determined by the adopted coding scheme.

The EHD at each time slot $T_0\leq t\leq T$, utilizes an energy of $\alpha(t)\cdot E(t)$ units to transmit $l(t) = \sqrt[m]{\frac{\alpha(t)g(t)E(t)}{\lambda}}$ bits to the AP. Note that $\alpha(t)$ depends on the channel gain and the residual battery level. In the subsequent slot, the battery evolves as $E(t+1) = (1-\alpha(t))E(t)$. The overall evolution of the energy state is as follows:
\begin{align}
E(t+1)=&\left\{
\begin{array}{ll}
E(t) + \eta g(t)P,& \text{if}\  1\leq t\leq T_0-1\\
(1-\alpha(t))E(t),& \text{if}\   T_0\leq t\leq T
\end{array}
\right.,\label{E_t}
\end{align}
where $\eta$ is a constant representing the efficiency of the energy harvesting process\footnote{Note that $\eta$ in practice is a function of the received power and cannot be assumed to be a constant. However, assuming a variable $\eta$ does not change the results of the paper. Thus, for ease of presentation, we assume that $\eta$ is constant.}. An illustrative example of the battery evolution is depicted in Figure \ref{fig:system}.

Our objective is to maximize the amount of data that can be transmitted over a duration of $T-T_0$ time slots by optimizing $T_0$ and $\alpha(E(t),g(t))$, for $t=T_0,\ldots,T$.

\section{Problem Formulation}
\label{PF}
%In this section, we derive an {\em online} low complexity algorithm that aims at maximizing the amount of bits that can be transmitted over any given deadline $T$.

In this section, we formulate a joint optimization problem that finds the optimal trade-off between the EH and IT periods, and the dynamic control of transmission power during the IT period. 
%
%Note that the EHD first needs to find an optimal stopping time for EH. When the optimal amount of energy has been harvested by the EHD, it will start the offloading process. Then, the EHD needs to optimally allocate its harvested energy to be able to offload the maximum amount of data until the deadline. 
More specifically, we aim at solving the following optimization problem.
\begin{align}
\max_{T_0,\{\alpha(t)\}^{T}_{t=T_0}} &\sum^{T}_{t=T_0} \mathbf{\sqrt[m]{\frac{\alpha(t)g(t)E(t)}{\lambda}}} \label{opt_problem}\\
&0\leq \{\alpha(t)\}^{T}_{t=T_0}\leq 1,\label{constraint1}
\end{align}
where $E(t)$ evolves as in (\ref{E_t}). Note that the objective function \eqref{opt_problem} is the total number of transmitted bits in the offloading period, \eqref{constraint1} ensures that the ratio of energy consumed does not exceed the available energy. Since $g(t)$ is only available causally, the optimization problem in \eqref{opt_problem}-\eqref{constraint1} cannot be solved using offline optimization tools and an online algorithm is required for its solution. 

%\subsubsection*{Remark} A common approach to solve similar problems is to use dynamic programming (DP) to find the solution numerically, and store the optimal decisions in a look-up table for the EHD. However, solving a DP and storing the result is prohibitive for resource constrained EHDs. In particular, note that in order to solve a DP,  the state space of the problem, $(E(t),g(t))$, for $t=1,\ldots,T$ should be discretized. Let the battery state be discretized using $L$ levels. Then, the table containing optimal decisions would require $L\times N\times T$ entries. Moreover, let the action space $\alpha(t)\in [0,\ 1]$ be discretized using $M$ levels. Then, in solving the DP, for every possible state $(E(t),g(t))$, we have to evaluate $M$ possible values of $\alpha(t)$ and choose the optimum value. Hence, even though the numerical solution of DP is possible, it is practically prohibitive due to its computational complexity and high storage costs.  In the following, we aim to obtain a low complexity solution by deriving closed form representations of the optimal decision variables.

\subsection{Dynamic Energy Allocation}
In this section, we first optimize the values of $\alpha(t)$ when $T_0$ is given. In the subsequent section, using the obtained result, we give a criteria for stopping the EH process, i.e., optimizing the value of $T_0$. 

Let the offloading period begin at $T_0$ and we aim to maximize the throughput over $T-T_0$ time slots by using DP. The problem can be solved by backwards recursion starting from the last state $t=T$.  Let the instantaneous reward of choosing $\alpha(t)$ be $U_{\alpha(t)}(E(t),g(t)) $ which is the instantaneous number of bits transmitted to the AP, when the the amount of available energy at time $t$, is $E(t)$, and the channel power gain is $g(t)$. Thus,
\begin{align}
U_{\alpha(t)}(E(t),g(t)) = \sqrt[m]{\frac{\alpha(t)g(t)E(t)}{\lambda}}.
\end{align}

We denote the action-value function by $V_{\alpha}(E(t),g(t))$ which is equal to the instantaneous reward of choosing $\alpha(t)$ plus the expected number of bits that can be transmitted in the future. Hence, the action-value function evolves as,
\begin{align}
V_{\alpha(t)}(E(t),g(t))= U_{\alpha(t)}(E(t),g(t)) + \sum^{N}_{n=1} q_n V(E(t+1),g_n), 
\end{align}
where $V(E(t),g(t))$ is the value function defined as,
\begin{align}
V(E(t),g(t)) =\max_{\alpha(t)} V_{\alpha(t)}(E(t),g(t)).
\end{align}

Note that at the last time slot, i.e., $t=T$, all the energy in the battery should be used for transmission, i.e., $\alpha(T) = 1$.  Thus, it follows that, 
\begin{align}
V(E(T),g(t))=&U_1(E(T),g(T)) \nonumber\\
=& \sqrt[m]{\frac{g(T)E(T)}{\lambda}}\nonumber\\
=&\sqrt[m]{\frac{g(T)(1-\alpha(T-1))E(T-1)}{\lambda}}.
\end{align}

% We start by calculating the optimal value of $\alpha(T-1)$ by conditioning on $g(T-1)$ and $E(T-1)$ as follows

% \begin{align}
% V(T)=U_1(T) =& \sqrt[m]{\frac{g(T)E(T)}{\lambda}}\nonumber\\
% =&\sqrt[m]{\frac{g(T)((1-\alpha(T-1))E(T-1)}{\lambda}}
% \end{align}
We maximize the action-value function at $t=T-1$ by optimizing $\alpha(T-1)$ as follows,
\begin{align}
V_{\alpha}(E(T-1),g(T-1)) =& U_{\alpha}(E(T-1),g(T-1)) \nonumber\\
&\hspace{-1cm}+ \sum^{N}_{n=1} q_n V((1-\alpha(T-1))E(T-1),g_n) \nonumber\\
&\hspace{-1cm}=\sqrt[m]{\frac{g(T-1)\alpha(T-1) E(T-1)}{\lambda}} \nonumber\\
&\hspace{-1cm}+ \sum^{N}_{n=1} q_n \sqrt[m]{\frac{g_n(1-\alpha(T-1))E(T-1))}{\lambda}}.\label{V_alpha_T_1}
\end{align}
It is easy to see that (\ref{V_alpha_T_1}) is a concave function of $\alpha(T-1)$. Therefore, using the first order optimality conditions on (\ref{V_alpha_T_1}), the optimal $\alpha(T-1)$ can be calculated as follows:
\begin{align}
\alpha^*(T-1) = \frac{g(T-1)^{\frac{1}{m-1}}}{g(T-1)^{\frac{1}{m-1}}+Q(T-1)^{\frac{m}{m-1}}}\label{alfa1},
\end{align}
where
\begin{align}
Q(T-1) = \sum^{N}_{n=1} q_n\sqrt[m]{g_n}.\label{Q1}
\end{align}
The corresponding value function can also be calculated as
\small
\begin{align}
&V(E(T-1),g(T-1))\nonumber\\
&\hspace{1cm}=\sqrt[m]{\frac{E(T-1)}{\lambda}}\big(g(T-1)^{\frac{1}{m-1}} + Q(T-1)^{\frac{m}{m-1}}\big)^\frac{m-1}{m}.\label{V1}
\end{align}
\normalsize

In a similar manner as above, we can recursively calculate the optimal $\alpha(t)$ for $t=T-2,\ldots,T_0$. The result is summarized in the following theorem. 

\begin{theorem}
\label{optimal}
For any $t = T-1, \ldots,T_0$, the optimal decision is to choose
\begin{align}
\alpha^*(t) = \frac{g(t)^{\frac{1}{m-1}}}{g(t)^{\frac{1}{m-1}}+Q(t)^{\frac{m}{m-1}}},\label{alfa_k}
\end{align}
where
\begin{align}
Q(t) = \sum^{N}_{n=1} q_n \big(g_n^{\frac{1}{m-1}} + Q(t+1)^{\frac{m}{m-1}}\big)^{\frac{m-1}{m}} \label{Q_k}.
\end{align}
The corresponding value function is 
\begin{align}
V(E(t),g(t)) = \sqrt[m]{\frac{E(t)}{\lambda}}\big(g(t)^{\frac{1}{m-1}} + Q(t )^{\frac{m}{m-1}}\big)^\frac{m-1}{m}\label{V_k}
\end{align}
\end{theorem}
\begin{proof}
The proof is by induction. The theorem is true for the base case, i.e., time slot $T-1$, as shown in (\ref{alfa1}), (\ref{Q1}), and (\ref{V1}). By assuming that (\ref{alfa_k}), (\ref{Q_k}), and (\ref{V_k}) is true for time slot $t+1$. The detailed proof can be found in \cite{arxiv}.
\end{proof}

Theorem \ref{optimal} gives a framework to dynamically allocate energy at each time slot $t\geq T_0$. The closed form expressions derived in (\ref{alfa_k})-(\ref{V_k}) significantly simplify the procedure to optimize $T_0$. We will use these to find an structure for the optimal stopping time for the EH period in the subsequent section.

\subsection{Optimal Stopping time for the EH Process}

In the following, we derive the optimal stopping time for the EH process, i.e., optimizing $T_0$ in \eqref{opt_problem}-\eqref{constraint1}. Recall that the EHD accumulates energy up to some time $t$, and then stops the EH process to start IT period. Also, recall that during EH, the EHD is blind to the channel conditions. If the EHD stops the EH process at time $t$, then the expected number of bits that can be transmitted is
\begin{align}
\sum^{N}_{n=1} q_n V(E(t),g_n) &= \sum^{N}_{n=1} q_n \sqrt[m]{\frac{E(t)}{\lambda}}\big(g_n^{\frac{1}{m-1}} + Q(t)^{\frac{m}{m-1}}\big)^\frac{m-1}{m}\nonumber\\
& = \sqrt[m]{\frac{E(t)}{\lambda}}Q(t-1). \label{V_avg}
\end{align}
Note that  (\ref{V_avg}) follows from (\ref{Q_k}).

Let $J_t(E(t))$, $t=1,\ldots,T$ be the maximum expected number of bits that can be transmitted if the EH process is stopped at time $t$, and the amount of available energy is $E(t)$. At any time $t$, the EHD will either stop the EH process or continue the process. The optimal stopping time for the EH process can be formulated as
\begin{align}
\max_{t\leq T}\,\, J_t(E(t)),
\end{align}
where
\begin{align}
J_t(E(t)) = \max\left(\sqrt[m]{\frac{E(t)}{\lambda}}Q(t-1), \mathbb{E}(J_{t+1}(E(t+1))|E(t))\right).
\end{align}

The problem can be formulated as a DP and solved for every possible $E(t)$ and $t$. Before proceeding, we need the following lemma.
\begin{lemma}
\label{lemma:Q_decreasing}
$Q(t)$, defined in (\ref{Q_k}) is a monotonically decreasing function in $t$.
\end{lemma}
\begin{proof}
By substituting $Q(t)$ from (\ref{Q_k}) into $\frac{Q(t)}{Q(t+1)}$, it can be shown that $Q(t)>Q(t+1)$.
\end{proof}

Note that at $t=T$, the best strategy is to stop the EH process and start IT, since otherwise no bits can be offloaded to the AP. Thus,
\begin{align}
J_T(E(T)) = \sqrt[m]{\frac{E(T)}{\lambda}} Q(T-1).
\end{align}

We continue the recursive evaluation at time slot $t =T-1$. We have,
\begin{align}
&J_{T-1}(E(T-1)) \nonumber\\
&= \max\Bigg(\sqrt[m]{\frac{E(T-1)}{\lambda}} Q(T-2),\mathbb{E}(J_T(E(T))|E(T-1))\Bigg)\nonumber\\
&= \max\Bigg(\sqrt[m]{\frac{E(T-1)}{\lambda}} Q(T-2),\nonumber\\
&\hspace{3cm}\sum^{N}_{n=1} q_n\sqrt[m]{\frac{E(T-1)+e_n}{\lambda}} Q(T-1)\Bigg),
\end{align}
where $e_n = \eta g_n P$ is the amount of harvested energy when the channel state is at level $n$. Since $Q(T-2)>Q(T-1)$ from Lemma \ref{lemma:Q_decreasing}, if $E(T-1)\geq \gamma(T-1)$ , then 
\begin{align}
\sqrt[m]{\frac{E(T-1)}{\lambda}} Q(T-2)\geq \sum^{N}_{n=1} q_n\sqrt[m]{\frac{E(T-1)+e_n}{\lambda}} Q(T-1)),
\end{align}
where $\gamma(T-1)$ is the solution to the following equation
\begin{align}
\sum^{N}_{n=1} q_n \sqrt[m]{1+\frac{e_n}{\gamma(T-1)}} = \frac{Q(T-2)}{Q(T-1)}. \label{gamma_T_1}
\end{align}
Note that $\gamma(T-1)$ admits a unique solution because the left hand side of (\ref{gamma_T_1}) is a strictly decreasing function in $\gamma(T-1)$ and its range belongs to $(1,\ \infty)$. Also, from Lemma \ref{lemma:Q_decreasing}, we know that $\frac{Q(T-2)}{Q(T-1)}>1$. 
Hence, it is optimal to stop the EH process at time $T-1$ if $E(T-1)\geq \gamma(T-1)$. This suggests that the optimal stopping times are governed by a time varying threshold type structure, where at any given time $t$, it is optimal to stop the EH process if $E(t)\geq\gamma(t)$.

In the following theorem, we give the structure of the optimal stopping policy.
\begin{theorem}
\label{theorem:threshold}
At each time slot $t$, the optimal decision is to stop the EH process if $E(t)\geq \gamma(t)$, where $\gamma(t)$ is the solution to the following equation,
\begin{align}
\sum^{N}_{n=1} q_n \sqrt[m]{1+\frac{e_n}{\gamma(t)}} = \frac{Q(t-1)}{Q(t)}\label{gamma}
\end{align}
\end{theorem}
\begin{proof}
The proof is by induction. We need to show that the result of the theorem is true for $J_{t}(E(t))$ for all $t=1,\ldots,T-1$. The result of the theorem is true for the base case of $t=T-1$ in (\ref{gamma_T_1}). Assume that the theorem holds for $t+1$, i.e., if $E(t+1)\geq\gamma(t+1)$, it is optimal to stop the EH process. Using this result, similar to (20)-(21), it can be shown that the case for time slot $t$ is also true. The detailed proof can be found in \cite{arxiv}.
\end{proof}
The results established in Theorem \ref{optimal} and \ref{theorem:threshold} enables us to develop an online low complexity optimal algorithm that maximizes the expected throughput. The procedure is summarized in Algorithm \ref{alg}. 
\begin{algorithm}
\caption{Optimal offloading algorithm}\label{alg}
\begin{algorithmic}[h]
\State Initialize $Q(t)$ for $t= 0,\ldots,T-1$ using (\ref{Q_k}),
\State Initialize $\gamma(t)$ for $t = 1,\ldots,T-1$ using (\ref{gamma}),
\For{$t = 1:T$}
\If{$E(t)<\gamma(t)$}
\State continue the EH process
\Else 
\State $T_0 = t$,
\State Stop the  EH process,
\State Break
\EndIf
\EndFor
\For{$t = T_0:T$}
\State Calculate $\alpha(t)$ using (\ref{alfa_k}),
\State Transmit using $\alpha(t)E(t)$.
\EndFor
\end{algorithmic}
\end{algorithm}

%%%%%%%%%%%%%%%%%%%%%%%%%%%%%%%%%%%%%%%%%%%%%%%%%%%%%%

\section{Numerical Results}
\label{numres}
\begin{figure}
        \centering
        \begin{subfigure}[h]{0.42\textwidth}
                \includegraphics[width=\textwidth]{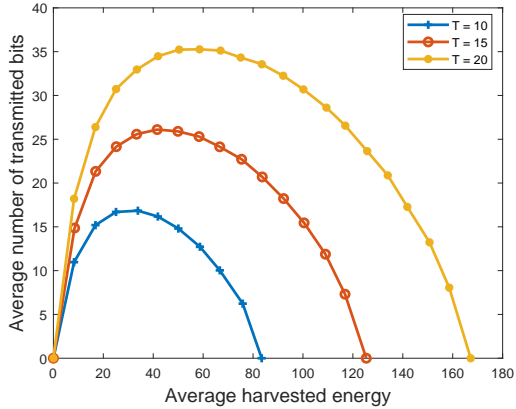}
                \caption{Energy-rate trade-off for different values of $T$.}
                \label{fig:R-E-T}
        \end{subfigure}
        \vfill
        \begin{subfigure}[h]{0.42\textwidth}
                \includegraphics[width=\textwidth]{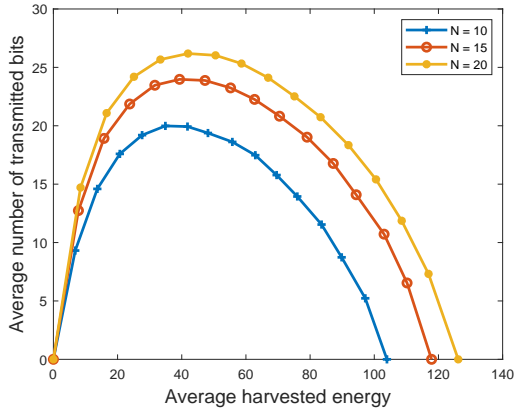}
                \caption{Energy-rate trade-off for different values of $N$. }
                \label{fig:R-E-N}
        \end{subfigure}
				\caption{Energy-rate trade-off.}
\label{fig:R-E}
\end{figure}

In this section, we first evaluate the rate-energy trade-off of the network which is the average number of bits transmitted with respect to the amount of harvested energy in a finite duration of $T$. In Figure \ref{fig:R-E-T}, for different values of $T$ the rate-energy trade-off is depicted. For different values of $N$, Figure \ref{fig:R-E-N}, illustrates the rate-energy trade-off. We observe from the figures that, spending too much time for transmitting more energy in the EH period reduces the time for IT period which in turn reduces the throughput. On the other hand, if we reduce the EH period, there would be less energy in the IT period resulting in a reduced throughput. Hence, an optimal balance is required.

\begin{figure}
        \centering
        \begin{subfigure}[h]{0.42\textwidth}
                \includegraphics[width=\textwidth]{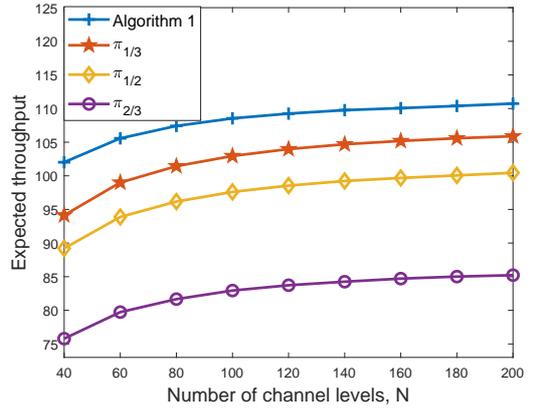}
                \caption{Expected throughput with respect to $N$.}
                \label{fig:TH_VS_N}
        \end{subfigure}
        \vfill
        \begin{subfigure}[h]{0.42\textwidth}
                \includegraphics[width=\textwidth]{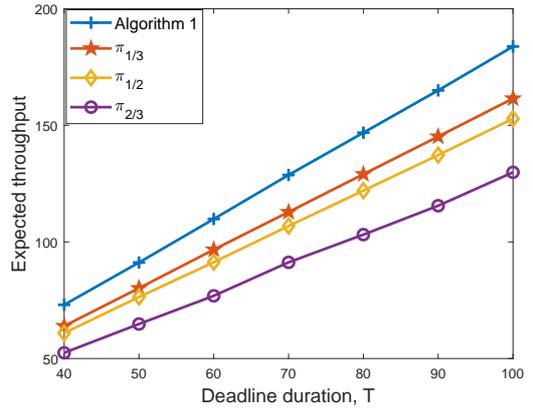}
                \caption{Expected throughput with respect to $T$. }
                \label{fig:TH_VS_T}
        \end{subfigure}
				\caption{The effect of channel discretization and deadline duration on the expected throughput.}
\label{cor_VS_R}
\end{figure}
\begin{figure}
        \centering
        \begin{subfigure}[h]{0.42\textwidth}
                \includegraphics[width=\textwidth]{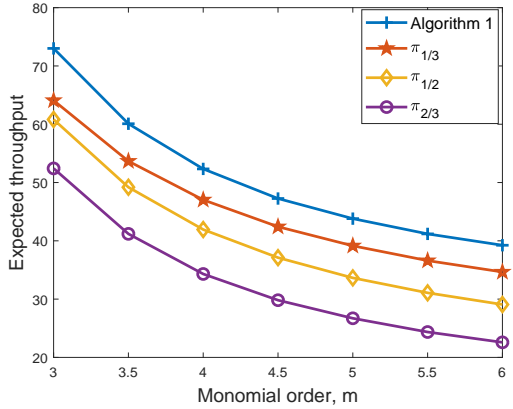}
                \caption{Expected throughput with respect to $m$.}
                \label{fig:TH_VS_m}
        \end{subfigure}
        \vfill
        \begin{subfigure}[h]{0.42\textwidth}
                \includegraphics[width=\textwidth]{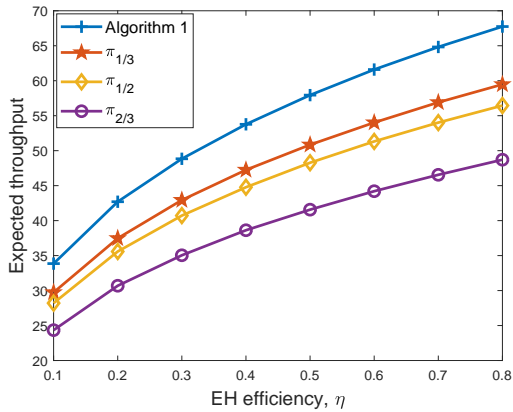}
                \caption{Expected throughput with respect to $\eta$ for $m=3$. }
                \label{fig:TH_VS_eta}
        \end{subfigure}
				\caption{The effect of the monomial order, $m$ and EH efficiency, $\eta$, on the expected throughput.}
\label{VS_m_eta}
\end{figure}

Next, we evaluate the performance of the optimal policy given in Algorithm \ref{alg} with respect to a simple policy denoted by $\boldsymbol\pi_{\beta}$. In policy $\boldsymbol\pi_{\beta}$, the EHD harvests energy for a duration of $\lfloor\beta\cdot T\rfloor$ time slots and utilizes the harvested energy uniformly in the remaining time slots for offloading its task until the deadline. The performance metric for evaluation is the expected throughput. For policy $\boldsymbol\pi_{\beta}$, throughout the simulation, we assume that $\beta \in\{ 1/3,1/2,2/3\}$.

We consider Rayleigh fading for the wireless channel and assume that $g(t)$ is exponentially distributed with mean $1$. We discretize $g(t)$, using $N$ equally spaced levels. In Figure \ref{fig:TH_VS_N}, we compare the performance of Algorithm \ref{alg} with  $\boldsymbol\pi_{\beta}$ by varying the number of discretization levels, $N$. We assume that $\lambda =0.1$, $m=3$, $P=10$, $T=50$. It can be seen that Algorithm \ref{alg} is able to outperform $\boldsymbol\pi_{\beta}$ for different values of $\beta$. An important observation from Figure \ref{fig:TH_VS_N} is that in order to achieve near-optimal performance, a sufficient number of discretization levels is required.  However, the computational complexity of numerically solving the DP quickly becomes prohibitively expensive as the number of discretization levels increase. On the contrary, increasing the discretization levels is not an issue for Algorithm \ref{alg} due to its lower computational complexity. 

Figure \ref{fig:TH_VS_T} illustrates the effect of the deadline duration, $T$, on the performance comparison of Algorithm \ref{alg} with $\boldsymbol\pi_{\beta}$. In this experiment, the number of channel discretization level is taken as $N=20$. As expected, increasing the deadline improves throughput, since more energy can be harvested and the EHD has a longer time to offload its task. It can be seen that, as the deadline duration increases, the gap between Algorithm \ref{alg} and $\boldsymbol\pi_{1/3}$, which is the best $\boldsymbol\pi_{\beta}$, also increases. 

The effect of the monomial order $m$, reflecting the adopted coding scheme, and EH efficiency ,$\eta$, on the expected throughput achieved by Algorithm 1 and $\boldsymbol\pi_\beta$ is depicted in Figure \ref{fig:TH_VS_m} and  \ref{fig:TH_VS_eta} for $\lambda =0.1$, $P=10$, $T=40$, $N=20$. It can be seen that Algorithm 1 outperforms $\boldsymbol\pi_\beta$ in both cases.

%  \begin{figure}[ht]
%   \centering
%     \includegraphics[scale=.55]{rate_energy.eps}
% 		  \caption{System model.}
% 			\label{fig:offload}
% \end{figure}

\section{Conclusions}
\label{CON}
In this work, we investigated the problem of finite horizon throughput maximization over a stochastic wireless channel when the deadline duration spans over multiple time slots with only causal CSI. We formulated the problem as a DP and by gaining insight into the DP, we reduced the dimension of the original from three to one enabling a closed form solution. By deriving closed form solutions for dynamic power allocation, and showing that the optimal stopping time for EH process follows a time varying threshold type structure, we developed a low complexity optimal algorithm, suitable for resource limited EHDs. As a future work, we will extend the results of the paper for the case of  multi-antenna APs and EHDs. Also, different performance metrics such as minimizing the task completion time and minimizing the power consumption of the AP will be addressed. 

\bibliographystyle{IEEEtran}
\bibliography{ref}

%%%%%%%%%%%%%%%%%%%

%\printbibliography[title={References}]
\end{document}